\documentclass[10pt,journal,doublecolumn]{IEEEtran}

\usepackage{amsmath,amssymb}
\usepackage[dvips]{graphicx}
\usepackage[usenames]{color}
\usepackage{amsfonts}
\usepackage{latexsym}

\usepackage[format=hang]{subfig}
\newtheorem{theorem}{\textit{{Theorem}}}

\newtheorem{definition}{\textit{{Definition}}}

\newtheorem{remark}{\textit{{Remark}}}
\newtheorem{proof}{\textit{{Proof}}}

\newtheorem{example}{\textit{{Example}}}

\begin{document}

\title{Sequence Design for Cognitive CDMA Communications under Arbitrary Spectrum Hole Constraint}

\author{Su~Hu, Zilong~Liu, Yong~Liang~Guan, Wenhui~Xiong, Guoan~Bi, Shaoqian~Li
\thanks{Manuscript received date: Jan 3, 2014; Manuscript revised date: May 5, 2014. The work of Su Hu is supported by the National Basic Research Program of China 2013CB329001, Natural Science Foundation of China 61101090/61101093 and Fundamental Research Funds for the Central Universities ZYGX2012Z004. The work of Zilong Liu and Yong Liang Guan is supported by
the Advanced Communications Research Program DSOCL06271 from the Defense Research and Technology Office (DRTech), Ministry of Defence, Singapore.}
\thanks{Su Hu, Wenhui Xiong, and Shaoqian Li are with National Key Laboratory on Communications, University of Electronic Science and Technology of China, Chengdu, China. E-mail: \{husu, whxiong, lsq\}@uestc.edu.cn.}
\thanks{Zilong Liu, Yong Liang Guan, and Guoan Bi are with the School of Electrical and Electronic Engineering, Nanyang Technological University, Singapore. E-mail: \{zilongliu, eylguan, egbi\}@ntu.edu.sg.}
}

\maketitle

\begin{abstract}
To support interference-free quasi-synchronous code-division multiple-access (QS-CDMA) communication with low spectral density profile in a cognitive radio (CR) network, it is desirable to design a set of CDMA spreading sequences with zero-correlation zone (ZCZ) property. However, traditional ZCZ sequences (which assume the availability of the entire spectral band) cannot be used because their orthogonality will be destroyed by the spectrum hole constraint in a CR channel. To date, analytical construction of ZCZ CR sequences remains open. Taking advantage of the Kronecker sequence property, a novel family of sequences (called ``quasi-ZCZ" CR sequences) which displays zero cross-correlation and near-zero auto-correlation zone property under arbitrary spectrum hole constraint is presented in this paper. Furthermore, a novel algorithm is proposed to jointly optimize the peak-to-average power ratio (PAPR) and the periodic auto-correlations of the proposed quasi-ZCZ CR sequences. Simulations show that they give rise to single-user bit-error-rate performance in CR-CDMA systems which outperform traditional non-contiguous multicarrier CDMA and transform domain communication systems; they also lead to CR-CDMA systems which are more resilient than non-contiguous OFDM systems to spectrum sensing mismatch, due to the wideband spreading.
\end{abstract}


\section{Introduction}
Cognitive radio (CR) is considered to be a promising paradigm to provide the capability of using or sharing the spectrum in an opportunistic manner to solve the scarcity of available spectrum. The spectrum opportunity is defined as spectrum holes that are not being used by the designated primary users at a particular time in a particular geographic area \cite{Haykin}-\cite{Budiarjo}. An interesting problem in CR systems is how to support interference-free multiple access communications with low spectral density profile and anti-jamming capability\footnote{In a tactical multi-user communication environment, for instance, an interference-free system with anti-detection and anti-jamming capabilities is critical for maintaining a robust communication quality in a battle field.}. To this end, a possible approach is to employ spread-spectrum code-division multiple-access (CDMA) using sequences with zero-correlation zone (ZCZ) property. Here, a ZCZ refers to a window of zero auto- and cross- correlations centered around the in-phase timing position. Owing to this correlation property, a quasi-synchronous CDMA (QS-CDMA) system employing ZCZ sequences as the CDMA spreading codes is able to achieve interference-free performance provided that all of the received signals fall into the ZCZ \cite{Suehiro94}-\cite{Tang06}. In the literature, ZCZ sequences were first proposed in \cite{Fan99}, and have been well developed by many researchers \cite{Tang01}-\cite{Tang10}.

However, traditional ZCZ sequences cannot be applied directly in CR systems. This is because their design generally assumes the availability of the entire spectral band (rather than certain non-contiguous spectral bands in a CR system as specified by the spectrum hole constraint) for every ZCZ sequence\footnote{From now on, a sequence which satisfies the spectrum hole constraint in a CR system is also called a ``CR sequence", to distinguish it from the traditional sequences with no spectrum hole constraint. }. The ZCZ property of a traditional ZCZ sequence set will be damaged/lost if a spectrum hole constraint is imposed by spectral nulling\footnote{We will show this in \textit{Example 1} in Section II.B.}. The same can be said for other traditional sequences with good correlation properties, e.g., polyphase sequences with ideal impulse-like auto-correlations \cite{Zadoff}-\cite{Mow95}. Recognizing this design challenge, a numerical approach was adopted by He \emph{et al} to design unimodular CR sequences with low out-of-phase auto-correlations in CR radar systems \cite{He10}. Their work was followed by Tsai \emph{et al} for CR sequences with low auto-correlations and low peak-to-average power ratio (PAPR), using convex optimization and Gerchberg-Saxton (GS) algorithm \cite{Tsai11}. However, the approaches in \cite{He10} and \cite{Tsai11} may not be applicable for the construction of CR sequences with low/zero cross-correlations. Till now, systematic construction of CR sequences with low/zero auto- and cross- correlations remains open, to the authors' best knowledge.

The first main contribution of this paper is an analytical construction of sequences with zero cross-correlation zones (ZCCZs), called ``quasi-ZCZ" CR sequences, for zero multiuser interference (MUI) CR-CDMA communications. A salient feature of our proposed construction is that the ZCCZ property holds for every distinct pair of quasi-ZCZ sequences regardless of any spectrum hole constraint (which may vary with time and location). Also, a large portion of the auto-correlations within the ZCCZ are zero except for certain small time-shifts near the in-phase timing position. We remark that the proposed CR sequences in this paper is different from the sequences in our recent work \cite{Hu13} which are for transform domain communication system (TDCS) based cognitive radio networks. On the other hand, the sequences in \cite{Hu13} don't possess ZCCZ centered around the in-phase timing position. The key idea of our proposed construction is to apply the Kronecker product to a ``seed" ZCZ sequence set and a CR waveform set which satisfies the spectrum hole constraint, thus enabling the use of the Kronecker sequence property \cite{Stark} to achieve the ZCCZ property under arbitrary spectrum hole constraint. The resultant sequences feature a two-dimensional structure in the time-frequency domains. Because of this, our proposed construction is also referred to as a ``time-frequency synthesis".

A second main contribution is a new numerical algorithm to suppress the PAPR and the out-of-phase periodic auto-correlations (in the ZCCZ) of the proposed CR sequences without destroying their ZCCZ property. Due to the property of the above-mentioned proposed construction, the task of this algorithm is equivalent to generating CR sequences with low PAPR and low out-of-phase aperiodic auto-correlations. Different from the CR sequences in \cite{He10} with some spectral leakage, our proposed algorithm takes a frequency-domain approach which leads to CR sequences with zero spectral leakage over the spectrum holes, in addition to their properties of low PAPR and low aperiodic auto-correlations. Also, unlike the algorithm in [\ref{Tsai11}, Section IV] which directly optimizes the aperiodic auto-correlations by solving a relaxed non-convex problem, our proposed algorithm makes use of the following sequence property: a sequence with low aperiodic auto-correlations also possesses low periodic auto-correlations. Simulations indicate that this algorithm features better suppression capabilities for both PAPR and aperiodic auto-correlations, compared with the algorithm in [\ref{Tsai11}, Section IV].

Thirdly, based on the proposed quasi-ZCZ CR sequences, a CR-CDMA scheme which consists of a maximum ratio combining Rake receiver is presented. We examine the resilience against multiuser interference (MUI) and spectrum sensing mismatch of the proposed CR-CDMA, and show that it is capable of achieving the single-user bit-error-rate (BER) performance in quasi-synchronous multipath fading channels. Note that in an OFDM system, every subcarrier modulates a distinct data stream, leading to high PAPR. This is different from our proposed CR-CDMA system where every data symbol is spread over all the available subcarrier channels. In this sense, our proposed CR-CDMA system may be classified as multicarrier CDMA with dynamic spectral nulling. Different from the PAPR reduction approaches used for OFDM \cite{Jiang}, our proposed CR-CDMA can achieve low PAPR by sequence optimization (as shown in Section III.B).

This paper is organized as follows. Section II gives some preliminaries and presents the sequence design problem in CR-CDMA systems. In Section III, we first introduce the analytical construction of the quasi-ZCZ CR sequences (each of which satisfies the spectrum hole constraint), then present the numerical algorithm for CR sequences with low PAPR and low aperiodic auto-correlation functions. In Section IV, the proposed CR-CDMA using the proposed quasi-ZCZ CR sequences is presented, followed by the numerical simulations in Section V. In this end, this paper is summarized in Section VI.

\section{Preliminaries and Problem Definition}
\subsection{Notations and Definitions}
The following notations will be used throughout this paper.

\begin{itemize}
\item [$-$] $\mathbf{X}^{\text{T}}$ and $\mathbf{X}^{\text{H}}$ denote the transpose and the Hermitian transpose of matrix $\mathbf{X}$, respectively;
\item [$-$] $\|\mathbf{A}\|$ denotes the Frobenius norm of matrix $\mathbf{A}$;
\item [$-$] $<\mathbf{a},\mathbf{b}>$ denotes the inner-product sum between two complex-valued sequences $\mathbf{a}$ and $\mathbf{b}$;
\item [$-$] $T^\tau({\mathbf{x}})$ denotes the left-cyclic-shift of $\mathbf{x}=[x_0,x_1,\cdots,x_{L-1}]^{\text{T}}$ for $\tau$ positions, i.e.,
\begin{displaymath}
T^\tau({\mathbf{x}})=[x_\tau,x_{\tau+1},\cdots,x_{L-1},x_0,x_1,\cdots,x_{\tau-1}]^{\text{T}};
\end{displaymath}
\item [$-$] ${Diag}[\mathbf{x}]$ denotes a diagonal matrix with the diagonal vector of $\mathbf{x}$ and all off-diagonal entries equal to 0;
\item [$-$] $\mathbf{1}_{m}$ and $\mathbf{0}_{n}$ denote a length-$m$ all-1 row-vector and a length-$n$ all-0 row-vector, respectively;
\item [$-$] $\mathcal{I}_N$ denotes the identity matrix of order $N$;
\item [$-$] Denote $\omega_N=\exp\left (\frac{\sqrt{-1}2\pi}{N} \right )$.
\end{itemize}

For two length-$L$ complex-valued sequences\footnote{For ease of presentation, all sequences in this paper are in the form of column vectors.}
\begin{displaymath}
\mathbf{a}=[a_0,a_1,\cdots,a_{L-1}]^{\text{T}},~~\mathbf{b}=[b_0,b_1,\cdots,b_{L-1}]^{\text{T}},
\end{displaymath}
let $C_{{\mathbf{a}},{\mathbf{b}}}\left(\tau\right)$ be the aperiodic cross-correlation function (ACCF) between $\mathbf{a}$ and $\mathbf{b}$, i.e.,
\begin{equation}
\begin{split}
C_{{\mathbf{a}},{\mathbf{b}}}\left(\tau\right) &= {\sum\limits_{n = 0}^{L - \tau - 1} {a_n {b^*_{n + \tau}} }},
\end{split}
\end{equation}
where $0 \leq \tau \leq L-1$. Also, let $R_{{\mathbf{a}},{\mathbf{b}}}\left(\tau\right)$ be the periodic cross-correlation function (PCCF) between $\mathbf{a}$ and $\mathbf{b}$, i.e.,
\begin{equation}\label{equ_periodicx}
R_{{\mathbf{a}},{\mathbf{b}}}\left(\tau\right) = \sum\limits_{n = 0}^{L - 1} {a_n} {b^*_{n + \tau} }=\left <\mathbf{a},T^\tau(\mathbf{b}) \right >,
\end{equation}
where the addition $n+\tau$ in (\ref{equ_periodicx}) is performed modulo $L$. Clearly, $R_{{\mathbf{a}},{\mathbf{b}}}\left(\tau\right)=R_{{\mathbf{a}},{\mathbf{b}}}^*\left(-\tau\right)$. In particular, when {$\mathbf{a}=\mathbf{b}$}, $C_{{\mathbf{a}},{\mathbf{a}}}(\tau)$ will be sometimes written as $C_{{\mathbf{a}}}(\tau)$ and called the aperiodic auto-correlation function (AACF). Similarly,  $R_{{\mathbf{a}},{\mathbf{a}}} (\tau)$ will be sometimes written as $R_{{\mathbf{a}}}(\tau)$ and called the periodic auto-correlation function (PACF).

Let $\mathcal{A}=\left\{\mathbf{a}_1,\mathbf{a}_2,\cdots,\mathbf{a}_{K}\right\}$ be a set of $K$ sequences, each of length $L$, i.e.,
\begin{eqnarray}
{{\bf{a}}_i} & = & \left[  {a^i_0, a^i_1, \cdots, a^i_n, \cdots, a^i_{L - 1}} \right ]^{\text{T}},~~1\leq i\leq K.
\end{eqnarray}
\begin{definition}
$\mathcal{A}$ is said to be a $(K,L,Z)$ zero-correlation zone (ZCZ) sequence set if and only if it satisfies the following two conditions:
\begin{enumerate}
\item $R_{{\mathbf{a}}_i}\left(\tau\right)=0$ holds for any $1\leq i\leq K$ and $1\leq |\tau|<Z$;
\item $R_{{\mathbf{a}}_i,{\mathbf{a}}_j}\left(\tau\right)=0$ holds for any $i\neq j$ and $0\leq |\tau|<Z$.
 \end{enumerate}
In addition, $\mathcal{A}$ is said to be a $(K,L,Z)$ ``quasi-ZCZ" sequence set with zero cross-correlation zone (ZCCZ) if and only if the second condition is satisfied.
\end{definition}


Consider the following two length-$N$ sequences
\begin{displaymath}
\mathbf{x}=\left[x_0, x_1, \cdots, x_{N-1} \right]^{\text{T}},~~\mathbf{y}=\left[ y_0, y_1, \cdots, y_{N-1} \right ]^{\text{T}},
\end{displaymath}
 and another two length-$L$ sequences below
 \begin{displaymath}
 \mathbf{d}=\left [d_0, d_1, \cdots, d_{L-1} \right]^{\text{T}},~~\mathbf{e}=\left[e_0, e_1, \cdots, e_{L-1} \right]^{\text{T}}.
 \end{displaymath}

\vspace{0.1in}

\begin{definition}
(Kronecker Sequence)
\begin{displaymath}
\mathbf{u}=[u_0,u_1,\cdots,u_{LN-1}]^{\text{T}}
\end{displaymath}
 and
\begin{displaymath}
 \mathbf{v}=[v_0,v_1,\cdots,v_{LN-1}]^{\text{T}}
\end{displaymath}
  are called two Kronecker sequences if
\begin{equation}
u_k  =  d_l \cdot x_n,~~v_k  =  e_l \cdot y_n,
\end{equation}
where $k=lN+n, 0 \leq l \leq L-1, 0 \leq n \leq N-1$. That is, $\mathbf{u} = \mathbf{d}  \otimes \mathbf{x}$ and $\mathbf{v} = \mathbf{e}  \otimes \mathbf{y}$, where $\otimes$ denotes the Kronecker product operation.
\end{definition}

\vspace{0.1in}

\begin{remark}
By \cite{Stark}, the PCCF $R_{\mathbf{u},\mathbf{v}}\left(\tau\right)$ between $\mathbf{u}$ and $\mathbf{v}$ can be written as
\begin{equation}\label{CCKRON}
R_{\mathbf{u},\mathbf{v}}\left(\tau\right) = R_{\mathbf{d},\mathbf{e}}\left(l\right) C_{\mathbf{x},\mathbf{y}}\left(n\right) + R_{\mathbf{d},\mathbf{e}}\left(l+1\right) C_{\mathbf{x},\mathbf{y}}\left(n-N\right),
\end{equation}
where $\tau=lN+n, 0 \leq l \leq L-1, 0 \leq n \leq N-1$.
\end{remark}

\subsection{Sequence Design Problem in CR-CDMA Systems}
 In a CR system, a spectrum opportunity is defined as the spectral bands which are not being used by the designated primary users at a particular time in a particular geographic area. In this paper, we assume that the entire spectrum is divided into $N$ subcarriers. We further assume a ``subcarrier marking vector" ${\mathbf{{S}}} = \left[ {S_0, S_1, \cdots, S_{N - 1}} \right]^\text{T}$ which gives the status of all subcarriers. For example, the value of $S_k$ is set to 1 if the $k$th subcarrier is available; otherwise, $S_k=0$. Denote by $\Omega$ the set of all unavailable subcarrier positions, i.e., $\Omega= \left\{ {k ~ |~S_k = 0} \right\}$. In this paper, $\Omega$ is also referred to as a ``spectrum hole constraint". A spectrum hole constraint is called non-trivial if $|\Omega|>0$. In this paper, a time-domain sequence which satisfies a spectrum hole constraint is also called a CR sequence.

Let $\left\{{\mathbf{B}}_i \right\}_{i=1}^K$ be a set of $K$ length-$N$ \textit{frequency-domain} sequences, each satisfying the spectrum hole constraint, i.e.,
\begin{eqnarray}
\label{PSET}
{{\bf{B}}_i} & = & \left[ {B^i_0, B^i_1, \cdots, B^i_k, \cdots, B^i_{N - 1}} \right]^\text{T},~~1 \leq i \leq K,
\end{eqnarray}
where $B^i_k=0$ if $k\in \Omega$. Denote by $\mathcal{F}_N=[f_{i,j}]_{i,j=0}^{N-1}$ the (scaled) discrete Fourier transform (DFT) matrix of order $N$, i.e.,
\begin{equation}
f_{i,j}=\frac{1}{\sqrt{N}} \omega^{-ij}_N,~~0\leq i,j\leq N-1.
\end{equation}
Note that $\mathcal{F}_N$ is a unitary matrix, i.e., $\mathcal{F}_N \mathcal{F}^{\text{H}}_N=\mathcal{I}_N$.
Thus, the (scaled) inverse discrete Fourier transform (IDFT) matrix of order $N$ is $\mathcal{F}^{\text{H}}_N$.

Given $\left\{{\mathbf{B}}_i \right\}_{i=1}^K$, a \textit{time-domain} sequence set $\left\{\mathbf{{b}}_i \right\}_{i=1}^K$ can be obtained by
\begin{displaymath}
\mathbf{{b}}_i=[b^i_{0},b^i_1,\cdots,b^i_n,\cdots,b^i_{N-1}]^{\text{T}}=\mathcal{F}^{\text{H}}_N\mathbf{B}_i.
\end{displaymath}

\begin{definition}
[PAPR of \textit{time-domain} sequence]

The PAPR of a \textit{time-domain} sequence $\mathbf{b}=[b_0,b_1,\cdots,b_n,\cdots,b_{N-1}]^{\text{T}}$ is defined as
\begin{equation}
\text{PAPR}(\mathbf{b})=\frac{\max\limits_{0\leq n \leq N-1}|b_n|^2}{(1/N)\sum\limits_{n=0}^{N-1}|b_n|^2}.
\end{equation}
By this definition, the lowest PAPR of $\mathbf{b}$ is equal to 1 (i.e., 0dB) if and only if all sequence elements have identical magnitude.
\end{definition}

\vspace{0.1in}

\begin{remark}
The sequence design problem in CR-CDMA systems is to design a set of \textit{time-domain} sequences $\left\{\mathbf{{b}}_i \right\}_{i=1}^K$ which satisfy the following conditions:
\begin{enumerate}
\item subject to a varying spectrum hole constraint (i.e., $\Omega$). Ideally, with zero spectral leakage over the spectrum holes;
\item with good auto- and cross- correlation property (e.g., ZCZ);
\item with low PAPR values. Ideally, PAPR=0dB.
\end{enumerate}
\end{remark}

\vspace{0.1in}

 Note that given a traditional ZCZ sequence set, whenever a non-trivial spectrum hole constraint is imposed, the ZCZ correlation property however cannot be guaranteed. We show this by the following example.

\vspace{0.1in}

\begin{example}
Let ${\mathbf{{S}}} = \left[ \mathbf{1}_{4},\mathbf{0}_{2},\mathbf{1}_{3},\mathbf{0}_{4},\mathbf{1}_{3}  \right]^{\text{T}}$
be the subcarrier marking vector. Therefore, the spectrum hole constraint is $\Omega=\{4,5\} \cup \{9,10,11,12\}$. Consider a $(K,L,Z)= (2,16,3)$ binary ZCZ sequence set (i.e., PAPR=0dB for each sequence) below.
\begin{displaymath}
\begin{split}
{\mathbf{z}_1}& = \left[1, 1, 1, 1, 1, -1, 1, -1, 1, 1, -1, -1, 1, -1, -1, 1 \right]^{\text{T}},\\
{\mathbf{z}_2}& = \left[1, -1, 1, -1, 1, 1, 1, 1, 1, -1, -1, 1, 1, 1, -1, -1\right]^{\text{T}}.
\end{split}
\end{displaymath}
To impose the spectrum hole constraint, the 16-point DFT is first applied to transform every ZCZ (\textit{time-domain}) sequence into a \textit{frequency-domain} sequence, followed by spectral nulling and the 16-point IDFT. As a result, the above ZCZ sequence set can be transformed to a new \textit{time-domain} sequence set $\{\mathbf{w}_1,\mathbf{w}_2\}$ which satisfies the spectrum hole constraint, i.e.,
\begin{displaymath}
\mathbf{w}_1= \mathcal{F}^{\text{H}}_{16} \cdot {Diag}[\mathbf{S}] \cdot \mathcal{F}_{16} \cdot \mathbf{z}_1,~~\mathbf{w}_2= \mathcal{F}^{\text{H}}_{16} \cdot {Diag}[\mathbf{S}] \cdot \mathcal{F}_{16} \cdot \mathbf{z}_2.
\end{displaymath}
The time-domain magnitudes, PACF and PCCF between $\{\mathbf{z}_1,\mathbf{z}_2\}$ and $\{\mathbf{w}_1,\mathbf{w}_2\}$ are shown in Fig. 1. For ease of presentation, these correlation values have been normalized.
\end{example}

\begin{figure}
\centering
  \includegraphics[width=3.2in]{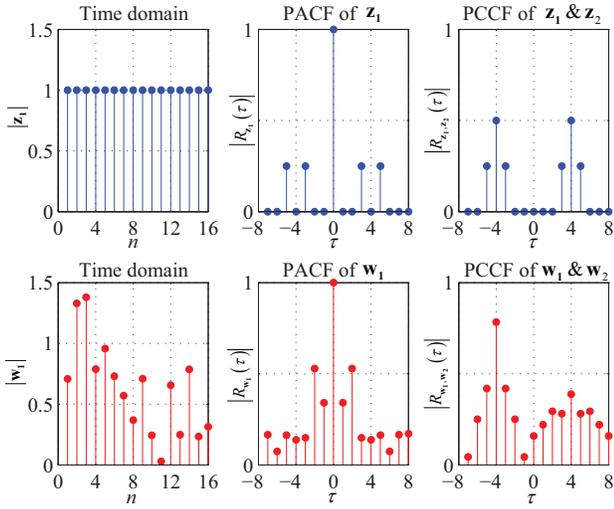}
  \caption{Time-domain magnitudes, PACF and PCCF between $\{\mathbf{z}_1,\mathbf{z}_2\}$ and $\{\mathbf{w}_1,\mathbf{w}_2\}$ }
  \label{FIG:LowZCZ0}
\end{figure}

\vspace{0.1in}

One can see that after spectral nulling, the correlation property of $\{\mathbf{w}_1,\mathbf{w}_2\}$ in \textit{Example 1} becomes unacceptable for supporting interference-free CR-CDMA communications due to the non-zero (non-trivial) PACFs and PCCFs within the ZCZ. In addition, the PAPR of $\mathbf{w}_1$ is increased to 5.53dB. This shows that new techniques are needed for the design of ZCZ CR sequences which satisfy the requirements given in \textit{Remark 2}.

\section{Construction of Quasi-ZCZ CR Sequence Sets}
In this section, we first present an analytical construction of quasi-ZCZ CR sequences which have ZCCZ for every distinct sequence pair, regardless of any spectrum hole constraint. We then introduce a new numerical algorithm to further optimize the PAPR and PACF of the proposed quasi-ZCZ CR sequences.

\vspace{0.1in}

\subsection{Proposed Quasi-ZCZ CR Sequences from the Time-Frequency Synthesis}
Let $\mathcal{A}=\left\{\mathbf{{a}}_i \right\}_{i=1}^K$  be a $\left(K,L,Z \right)$ ZCZ sequence set, where
\begin{eqnarray}
{{\bf{a}}_i} & = & \left[ {a^i_0, a^i_1, \cdots, a^i_l, \cdots, a^i_{L - 1}} \right]^{\text{T}},~~1\leq i \leq K.
\end{eqnarray}
Also, let $\mathcal{B}=\{\mathbf{b}_i\}_{i=1}^{K}$ be a set of length-$N$ \textit{time-domain} sequences obtained from the IDFT of $\left\{{\mathbf{B}}_i \right\}_{i=1}^K$ (a set of {\textit{frequency-domain} sequences} subject to a spectrum hole constraint $\Omega$), i.e., $\mathbf{b}_i=\mathcal{F}^{\text{H}}_N\mathbf{B}_i$. Construct a sequence set $\mathcal{C}=\left\{\mathbf{c}_i \right\}_{i=1}^K$ (each sequence of length $NL$) as follows,
\begin{equation}\label{equ_propose_cons}
\mathbf{c}_i=\mathbf{a}_i \otimes \mathbf{b}_i=\mathbf{a}_i \otimes \mathcal{F}^{\text{H}}_N \mathbf{B}_i,~~1\leq i\leq K.
\end{equation}

\begin{theorem}
$\mathcal{C}$ is a $(K,NL,NZ-N)$ ``quasi-ZCZ" CR sequence set, i.e., for any $i\neq j$,
 \begin{equation}
 \label{NewCC}
R_{{{\bf{c}}_i},{{\bf{c}}_j}}\left(\tau\right) = 0,~~|\tau| < NZ-N.
\end{equation}
In addition, each sequence in $\mathcal{C}$ satisfies the spectrum hole constraint $\Omega$ and with
\begin{equation}
\label{NewAC}
R_{\mathbf{c}_i}\left(\tau\right) = \left\{ {\begin{array}{*{20}{c}}
L \cdot C_{\mathbf{b}_i}\left( \tau  \right), &|\tau| < N;\\
0, &N \leq |\tau| < NZ-N.
\end{array}} \right.
\end{equation}
\end{theorem}

\begin{proof}
By \eqref{CCKRON}, the PCCF between $\mathbf{c}_i$ and $\mathbf{c}_j$ can be written as
\begin{equation}\label{ACdesign}
\begin{split}
R_{{{\bf{c}}_i},{{\bf{c}}_j}}\left(\tau\right) & = R_{{{\bf{a}}_i},{{\bf{a}}_j}}\left(l \right) C_{\mathbf{b}_i,\mathbf{b}_j}\left( n \right) \\
& ~~~~~+ R_{{{\bf{a}}_i},{{\bf{a}}_j}}\left(l+1 \right) C_{\mathbf{b}_i,\mathbf{b}_j}\left( n-N\right),
\end{split}
\end{equation}
where $l=\lfloor \tau/N\rfloor,n=(\tau~\text{mod}~N)$. For $|\tau| < NZ-N$, we have $|l|<|l+1|\leq Z-1$ and therefore, by the zero cross-correlation zone property of $\left\{\mathbf{{a}}_i \right\}_{i=1}^K$,
\begin{equation}
R_{{{\bf{a}}_i},{{\bf{a}}_j}}\left(l \right)=R_{{{\bf{a}}_i},{{\bf{a}}_j}}\left(l+1 \right)=0,~~i\neq j.
\end{equation}
Hence, for $i\neq j$, we have $R_{{{\bf{c}}_i},{{\bf{c}}_j}}\left(\tau\right)=0$ for $|\tau| < NZ-N$. It follows that $\mathcal{C}$ is a sequence set of ZCCZ width of $NZ-N$.

On the other hand, if $i=j$, (\ref{ACdesign}) is reduced to
\begin{equation}
R_{{{\bf{c}}_i}}\left(\tau\right) = R_{{{\bf{a}}_i}}\left(l \right) C_{\mathbf{b}_i}\left( n \right) + R_{{{\bf{a}}_i}}\left(l+1 \right) C_{\mathbf{b}_i}\left( n-N\right).
\end{equation}
By the zero auto-correlation zone property of $\left\{\mathbf{{a}}_i \right\}_{i=1}^M$, one can readily show that (\ref{NewAC}) is true. In the end, since each sequence can be written as
\begin{equation}\label{equ_ci}
\mathbf{c}^{\text{T}}_i=\Bigl [a^i_0 (\mathcal{F}^{\text{H}}_N\mathbf{B}_i)^{\text{T}},a^i_1 (\mathcal{F}^{\text{H}}_N\mathbf{B}_i)^{\text{T}},\cdots,a^i_{L-1} (\mathcal{F}^{\text{H}}_N\mathbf{B}_i)^{\text{T}} \Bigl ],
\end{equation}
where $\mathbf{B}_i$ is the $i$th {\textit{frequency-domain} sequence} subject to $\Omega$. It follows that each sequence in $\mathcal{C}$ also satisfies the spectrum hole constraint $\Omega$, thus completing the proof.
\end{proof}

\vspace{0.1in}

Note that the idea of the proposed construction is to assign every user a specific ZCZ sequence which is spread over all available subcarrier channels, followed by the inverse Fourier transform specified by the corresponding frequency-domain sequence $\mathbf{B}_i$. Removing away the IDFT matrix $\mathcal{F}_N$ in (\ref{equ_ci}), every sequence (say, $\mathbf{c}_i$) can be decomposed in time- and frequency- domains as shown in Fig.~\ref{FIG:FDNulling}. For instance, $a^i_l B^i_m$ is assigned as the sequence element at the $l$th time-slot and the $m$th subcarrier channel, thus forming a time-frequency lattice of $\mathbf{c}_i$. Because of this, the proposed construction is sometimes also referred to as a ``time-frequency synthesis". To show the spectrum hole constraint specified by $\Omega$, these columns corresponding to the unavailable spectral bands (occupied by the primary users) are nulled (denoted as ``0" in Fig.~\ref{FIG:FDNulling}).

\begin{figure}
\centering
  \includegraphics[width=3.2in]{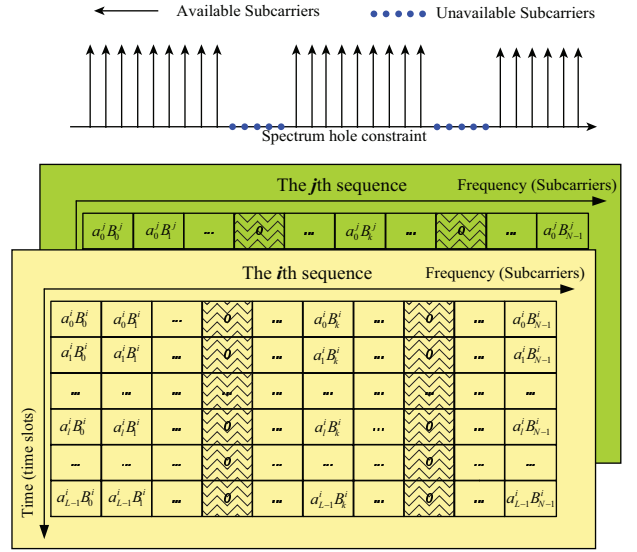}
  \caption{Time-frequency lattice of $\mathbf{c}_i$ and $\mathbf{c}_j$.}
  \label{FIG:FDNulling}
\end{figure}

Furthermore, we have the following remarks.
\begin{enumerate}
\item Different from the CR sequences (i.e., $\{\mathbf{w}_1,\mathbf{w}_2\}$) in \textit{Example 1} (where every CR sequence consists of one \textit{frequency-domain} sequence only), each of our proposed quasi-ZCZ CR sequences is a concatenation of multiple \textit{frequency-domain} sequences.
  \item The strength of our proposed construction is that the resultant sequence set has ZCCZ (which enables interference-free CR-CDMA communications), regardless of any spectrum hole constraint (thanks to the Kronecker sequence property). Therefore, the barrier of the spectrum hole constraint, which prevents conventional sequences with good correlations from finding applications in CR-CDMA systems, is fixed.
  \item To support more CR-CDMA users, it is desirable to employ a bound-achieving ``seed" ZCZ sequence set, i.e., $\left\{\mathbf{{a}}_i \right\}_{i=1}^K$ with sequence parameters of $(K,L,Z)$, for the construction of the proposed $(K,NL,NZ-N)$ quasi-ZCZ CR sequence set. For a polyphase ZCZ sequence set, it is known that $K\leq \lfloor L/Z \rfloor $ \cite{Tang00}. A bound-achieving polyphase ZCZ sequence set constructed from the generalized Chirp-like sequences was presented in \cite{Popovic10}. However, for a binary ZCZ sequence set, it is widely believed that $K\leq \left \lfloor \frac{L}{2(Z-1)} \right \rfloor$ where $Z\geq 3$ \cite{Tang10}. We will show that the ``seed" binary ZCZ sequence set in \textit{Example 2} satisfies the latter set size upper bound.
\end{enumerate}

\vspace{0.1in}

\begin{example}
Let ${\mathbf{{S}}} = \left[ \mathbf{1}_{14},\mathbf{0}_{6},\mathbf{1}_{20},\mathbf{0}_{8},\mathbf{1}_{16}  \right]^{\text{T}}$ be the subcarrier marking vector. Thus,
\begin{equation}
\Omega = \{14,15,\cdots,19\}\cup \{40,41,\cdots,47\} .
\end{equation}
Define $\{\mathbf{B}_i\}_{i=1}^{K}$ with $K=4$ and sequence length $N=64$, where
\begin{equation}
\begin{split}
\mathbf{B}_i & =[B^i_0,B^i_1,\cdots,B^i_k,\cdots,B^i_{N-1}]^{\text{T}}\\
B^i_k &  =
\left \{
\begin{array}{ll}
\exp\left (\frac{-\sqrt{-1} \pi u_i k^2}{N}\right ),& k \in \{0,1,\cdots,63\} \setminus \Omega,\\
0,& k\in \Omega
\end{array}
\right .
\end{split}
\end{equation}
and $(u_1,u_2,u_3,u_4)=(3,5,7,9)$. Note that if the spectrum hole constraint is removed, each $\mathbf{B}_i$ will become a Zadoff-Chu sequence\footnote{Zadoff-Chu sequences are used here for easy repetition of the proposed quasi-ZCZ CR sequences. In practice, $\{\mathbf{B}_i\}_{i=1}^{K}$ can be random sequences which satisfy the spectrum hole constraint.}. The time- and frequency- domain magnitudes, and the AACF of $\mathbf{b}_1=\mathcal{F}^{\text{H}}_{64} \mathbf{B}_1$ are shown in Fig. 3-a. It turns out that the PAPR (i.e., equals to 4.1dB) and the maximum out-of-phase AACF (i.e., equals to 0.2131) may be unacceptable for practical applications. As such, we will present an algorithm in Section III.B to further optimize the PAPR and the AACF.

Consider a $(K,L,Z)= (4,16,3)$ binary ZCZ sequence set below
\begin{displaymath}
\begin{split}
{{\bf{a}}_1}& = \left[1, 1, -1, 1, 1, 1, -1, 1, -1, -1, -1, 1, -1, -1, -1, 1 \right]^{\text{T}},\\
{{\bf{a}}_2}& = \left[-1, -1, 1, -1, 1, 1, -1, 1, 1, 1, 1, -1, -1, -1, -1, 1\right]^{\text{T}}, \\
{{\bf{a}}_3}& = \left[1, -1, -1, -1, 1, -1, -1, -1, -1, 1, -1, -1, -1, 1, \right. \\
            &    ~~~~~~~~~~~~~~~~~~~~~~~~~~~~~~~~~~~~~~~~~~~~~~~~~~~~~~~ \left. -1, -1\right]^{\text{T}}, \\
{{\bf{a}}_4}& = \left[-1, 1, 1, 1, 1, -1, -1, -1, 1, -1, 1, 1, -1, 1, -1, -1\right]^{\text{T}}.
\end{split}
\end{displaymath}
 Applying our proposed construction in (\ref{equ_propose_cons}), a $(K,NL,NZ-N)=(4,1024,128)$ ``quasi-ZCZ" CR sequence set $\left\{\mathbf{{c}}_i \right\}_{i=1}^K$ is obtained. To see this, the PACF of $\mathbf{c}_4$ and the PCCF between $\mathbf{c}_1$ and $\mathbf{c}_4$ are shown in Fig. 3-b. One can see that $\mathbf{c}_1$ and $\mathbf{c}_4$ have ZCCZ of width 128, thus \eqref{NewCC} is verified. Similarly, one can verify \eqref{NewAC}.
\end{example}

\begin{figure*}
\centerline{
\subfloat[Time- and frequency- domain magnitudes, and AACF of $\mathbf{b}_4$.]{\includegraphics[width=3.2in]{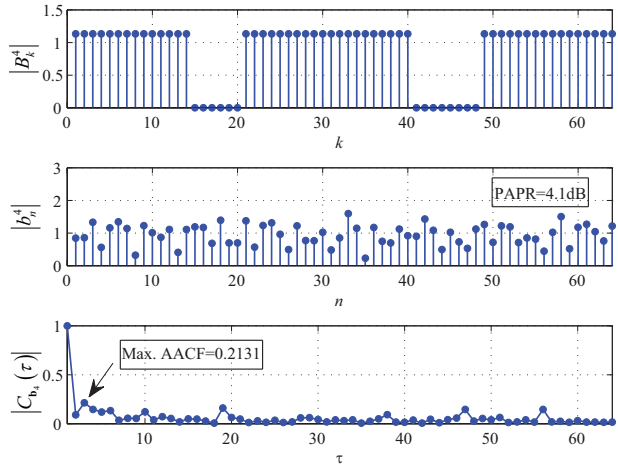}
\label{fig_examp2_a}}
\hfil
\subfloat[PACF of $\mathbf{c}_4$ and PCCF between $\mathbf{c}_1$ and $\mathbf{c}_4$.]{\includegraphics[width=3.2in]{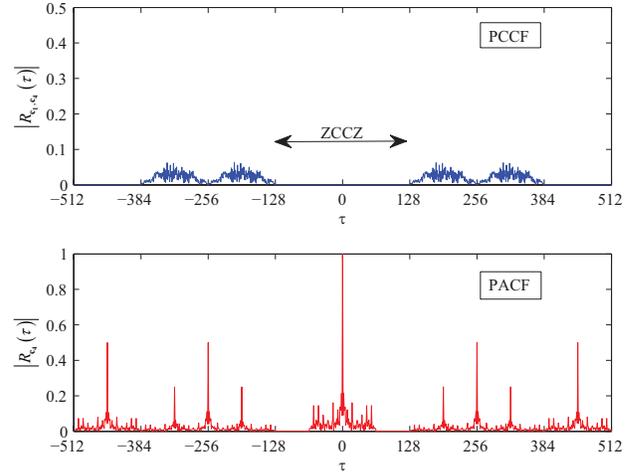}
\label{fig_examp2_b}}
}
\caption{Correlation plots for the quasi-ZCZ CR sequences constructed in \textit{Example 2}.}
\label{FIG:LowZCZ}
\end{figure*}

\subsection{Proposed Algorithm for CR Sequences with Low PAPR and Low AACF}
In this subsection, we present a numerical algorithm to further optimize the proposed quasi-ZCZ CR sequences in Section III.A. An interesting observation is that the correlation properties shown in \textit{Theorem 1} will also hold\footnote{Except for the out-of-zone auto- and cross- correlations for $|\tau|\geq NZ-N$.} if setting all sequences in $\mathcal{B}$ to be identical (say, all equal to $\mathbf{b}=[b_0,b_1,\cdots,b_n,\cdots,b_{N-1}]^{\text{T}}$). In this case, the PACF of all sequences for $|\tau|<N$ will be fully determined by the AACF of $\mathbf{b}$ (i.e., $C_{\mathbf{b}}(\tau)$), as shown in (\ref{NewAC}). Also, the CR sequence $\mathbf{b}_1$ (shown in Fig. 3-a) in \textit{Example 2} has a large PAPR value of 4.10dB which is not preferable for higher energy transmission efficiency. 
These observations motivate us to search for a single CR sequence with low PAPR and (non-trivial) low AACF over $|\tau|<N$.

For simplicity, we suppose $\|\mathbf{b}\|^2_2=N$, i.e., $\sum_{n=0}^{N-1}|b_n|^2=N$. By Parseval's theorem, we have
\begin{equation}
\|\mathbf{B}\|^2_2=\mathbf{b}^{\text{H}}\mathcal{F}^{\text{H}}_N\mathcal{F}_N \mathbf{b}=N,
\end{equation}
where $\mathbf{B}=[B_0,B_1,\cdots,B_k,\cdots,B_{N-1}]^{\text{T}}$. To optimize the AACF of $\mathbf{b}$, we recall the key idea of the CAN (cyclic algorithm new) in \cite{Stoica09} proposed by Stoica \emph{et al}: the correlation sidelobes of $\mathbf{b}$ vanish if the $2N$-point DFT of $[\mathbf{b}^{\text{T}},\mathbf{0}_N]^{\text{T}}$ have identical magnitude of $1/\sqrt{2}$. Hence, a CR sequence $\mathbf{b}$ with low AACF sidelobes can be obtained by solving the following problem:
\begin{equation}
\begin{split}
 \min\limits_{\mathbf{B},\mathbf{P}} &~~\mathcal{J}_1(\mathbf{B},\mathbf{P})=\left \|\mathcal{F}_{2N} \left [ \begin{matrix} \mathcal{F}^{\text{H}}_{N}\mathbf{B}\\ \mathbf{0}^{\text{T}}_N \end{matrix}\right ]  - \mathbf{P}\right \|^2_2,\\
\text{s.t.}  & ~~(1): B_k=0~\text{if}~k\in \Omega;\\
             & ~~(2): |P_k|=\frac{1}{\sqrt{2}},~~k=0,1,\cdots,2N-1,
\end{split}
\end{equation}
where $\mathbf{P}=[P_0,P_1,\cdots,P_{2N-1}]^{\text{T}}$.

On the other hand, to optimize the PAPR of $\mathbf{b}$, we wish to solve the following problem:
\begin{equation}
\begin{split}
 \min\limits_{\mathbf{B},\mathbf{p}} &~~\mathcal{J}_2(\mathbf{B},\mathbf{p})=\left \| \mathcal{F}^{\text{H}}_{N}\mathbf{B}  - \mathbf{p}\right \|^2_2,\\
\text{s.t.}  & ~~(1): B_k=0~\text{if}~k\in \Omega;\\
             & ~~(2): |p_k|=1,~~k=0,1,\cdots,N-1,
\end{split}
\end{equation}
where $\mathbf{p}=[p_0,p_1,\cdots,p_{N-1}]^{\text{T}}$.

By introducing a penalty factor $\lambda\in [0,1]$ which controls the relative weighting of $\mathcal{J}_1$ and $\mathcal{J}_2$, our optimization problem can now be formulated as follows:
\begin{equation}\label{mini_PAPR_ACF}
\begin{split}
 \min\limits_{\mathbf{B},\mathbf{P},\mathbf{p}} &~~\mathcal{J}(\mathbf{B},\mathbf{P},\mathbf{p})=\lambda \mathcal{J}_1(\mathbf{B},\mathbf{P}) + (1-\lambda) \mathcal{J}_2(\mathbf{B},\mathbf{p}), \\
\text{s.t.}  & ~~(1): B_k=0~\text{if}~k\in \Omega;\\
             & ~~(2): |P_k|=\frac{1}{\sqrt{2}},~~k=0,1,\cdots,2N-1;\\
             & ~~(3): |p_k|=1,~~k=0,1,\cdots,N-1.
\end{split}
\end{equation}

\vspace{0.1in}

To solve the minimization problem in (\ref{mini_PAPR_ACF}), our proposed algorithm is given as follows.
 \begin{enumerate}
 \item First, find the power spectrum of $\mathbf{b}$, i.e., ${\text{\mbox{\boldmath{$\beta$}}}}=[|B_0|^2,|B_1|^2,\cdots,|B_{N-1}|^2]^{\text{T}}$. Note that a necessary condition for $\mathbf{b}$ to have uniformly low AACF is that it should also have uniformly low PACF. Therefore, we consider to generate a CR sequence (say, $\hat{\mathbf{b}}$) with uniformly low PACF and apply the power spectrum of $\hat{\mathbf{b}}$ as a suboptimal solution of ${\text{\mbox{\boldmath{$\beta$}}}}$. By utilizing {the Wiener-Khinchin Theorem} that the power spectrum and the PACF form a Fourier transform pair, as shown in \cite{Tsai11}, we have
     \begin{equation}\label{PwrSpecFFT}
     R_{\mathbf{b}}(\tau)=\sqrt{N}\mathcal{F}_N(\tau,:) {\text{\mbox{\boldmath{$\beta$}}}},~~\tau\in \{1,2,\cdots,N-1\},
     \end{equation}
     where $\mathcal{F}_N(\tau,:)$ denotes the $\tau$th row sequence of $\mathcal{F}_N$.

\vspace{0.1in}

     Thus, our task can be transformed to solving the following mini-max problem:
     \begin{equation}\label{minimax_PACF}
     \begin{split}
 \min_{{\text{\mbox{\boldmath{$\beta$}}}}} & \max\limits_{1\leq \tau \leq N-1}  \left | \mathcal{F}_N(\tau,:) {\text{\mbox{\boldmath{$\beta$}}}} \right |\\
\text{s.t.} &  ~~(1): B_k=0~\text{if}~k\in \Omega;\\
            &  ~~(2): \sum \limits_{k=0}^{N-1}|B_k|^2=N.
\end{split}
\end{equation}
 \item Then, apply the Gerchberg-Saxton (GS) algorithm \cite{GS72} to find the optimal phases of $\mathbf{B},\mathbf{P},\mathbf{p}$, denoted by
 \begin{displaymath}
 \begin{split}
\psi_{\mathbf{B}} & =[\psi^{\mathbf{B}}_0,\psi^{\mathbf{B}}_1,\cdots,\psi^{\mathbf{B}}_{N-1}]^{\text{T}},\\
\psi_{\mathbf{P}} & =[\psi^{\mathbf{P}}_0,\psi^{\mathbf{P}}_1,\cdots,\psi^{\mathbf{P}}_{2N-1}]^{\text{T}},\\
\psi_{\mathbf{p}} & =[\psi^{\mathbf{p}}_0,\psi^{\mathbf{p}}_1,\cdots,\psi^{\mathbf{p}}_{N-1}]^{\text{T}}
\end{split}
 \end{displaymath}
  respectively. Note that the GS algorithm is guaranteed to converge and has been widely used in \cite{Tsai11}, \cite{Stoica09} and \cite{Stoica09-2}. 
  \begin{enumerate}
  \item Initialize the values of $\psi_{\mathbf{B}},\psi_{\mathbf{P}},\psi_{\mathbf{p}}$ which can be, say, some random numbers over $[0,2\pi)$. Also, initialize $\lambda$ depending on the priorities of PAPR and AACF, e.g., $\lambda$ is set to be a larger value over $[0,1]$ if AACF has a higher priority.
  \item Fix $\mathbf{B}$, choose
  \begin{equation}\label{num_alg_stp1}
  \psi_{\mathbf{P}}=\arg \left \{ \mathcal{F}_{2N} \left [ \begin{matrix} \mathcal{F}^{\text{H}}_{N}\mathbf{B}\\ \mathbf{0}^{\text{T}}_N \end{matrix}\right ]\right \}~\text{and}~\psi_{\mathbf{p}}=\arg \Bigl \{ \mathcal{F}^{\text{H}}_{N}\mathbf{B} \Bigl \}
  \end{equation}
  for minimum value of $\mathcal{J}$ in (\ref{mini_PAPR_ACF}).
  \item Fix $\mathbf{P}$ and $\mathbf{p}$, choose
    \begin{equation}\label{num_alg_stp2}
\psi_{\mathbf{B}}=\arg \Bigl \{ \lambda  \mathcal{F}_{N}\hat{\mathbf{p}} + (1-\lambda)\mathcal{F}_{N}\mathbf{p} \Bigl \},
  \end{equation}
  where $\hat{\mathbf{p}}$ denotes the first $N$ elements of $\mathcal{F}^{\text{H}}_{2N}\mathbf{P}$, for minimum value of $\mathcal{J}$ in (\ref{mini_PAPR_ACF}). {This is because finding the minimum of $\mathcal{J}$ in this case is equivalent to minimizing the following term.
  \begin{displaymath}
  \mathcal{J}^{(1)}=\lambda \left \| \left [ \begin{matrix} \mathcal{F}^{\text{H}}_{N}\mathbf{B}-\hat{\mathbf{p}} \\ \check{\mathbf{p}} \end{matrix}\right ]  \right \|^2_2 + (1-\lambda) \left \| \mathcal{F}^{\text{H}}_{N}\mathbf{B}  - \mathbf{p}\right \|^2_2,
  \end{displaymath}
  where $\check{\mathbf{p}}$ denotes the second $N$ elements of $\mathcal{F}^{\text{H}}_{2N}\mathbf{P}$. Since $\mathbf{P}$ is fixed (so is $\check{\mathbf{p}}$), finding the minimum of $\mathcal{J}^{(1)}$ is equivalent to minimizing
    \begin{displaymath}
  \mathcal{J}^{(2)}=\lambda \left \| \mathcal{F}^{\text{H}}_{N}\mathbf{B}-\hat{\mathbf{p}}  \right \|^2_2 + (1-\lambda) \left \| \mathcal{F}^{\text{H}}_{N}\mathbf{B}  - \mathbf{p}\right \|^2_2.
  \end{displaymath}
  $\psi_{\mathbf{B}}$ in (\ref{num_alg_stp2}) follows by solving $\min_{\mathbf{B}}{ \mathcal{J}^{(2)}}$.
  }

  \end{enumerate}
  Repeat steps b) and c) iteratively until a pre-set termination condition is met, e.g.,
  \begin{displaymath}
  \left \| \mathbf{B}^{(i)}-\mathbf{B}^{(i-1)}\right \|<\varepsilon=10^{-5},
  \end{displaymath}
  where $\mathbf{B}^{(i)}$ denotes the frequency-domain sequence $\mathbf{B}$ obtained after the $i$th iteration.
 \end{enumerate}

The flow chart of our proposed algorithm is given in Fig. 4.
\begin{figure}
\centering
  \includegraphics[width=2in]{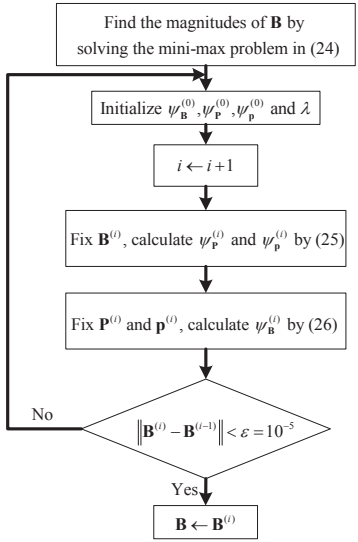}
  \caption{Flow chart of our proposed algorithm for CR sequences with low PAPR and low AACF.}
  \label{fig_flowchart}
\end{figure}

\begin{figure*}
\centerline{
\subfloat[$\lambda=0.15$]{\includegraphics[width=3.4in]{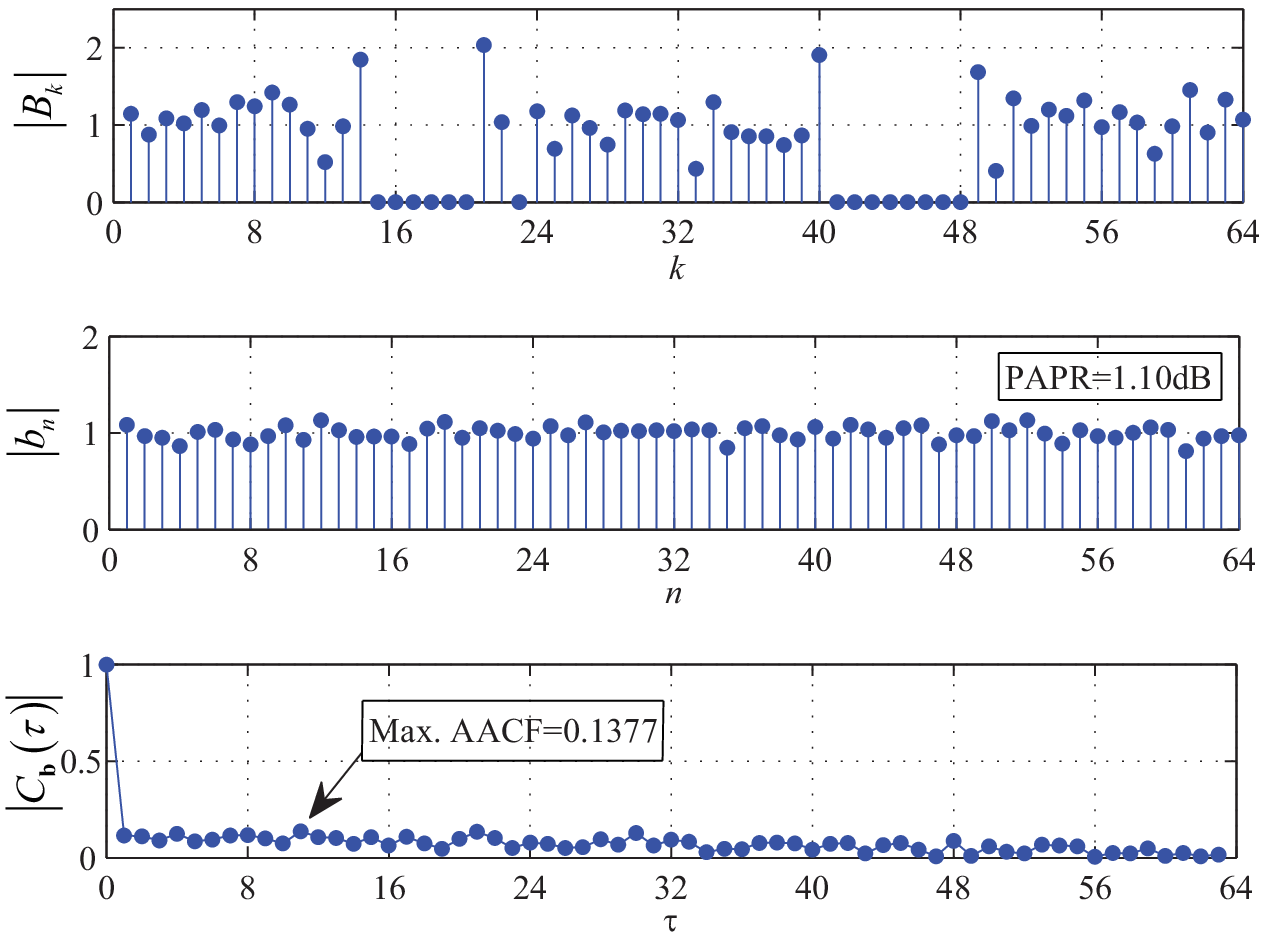}
\label{fig_examp3_a}}
\hfil
\subfloat[$\lambda=0.95$]{\includegraphics[width
=3.4in]{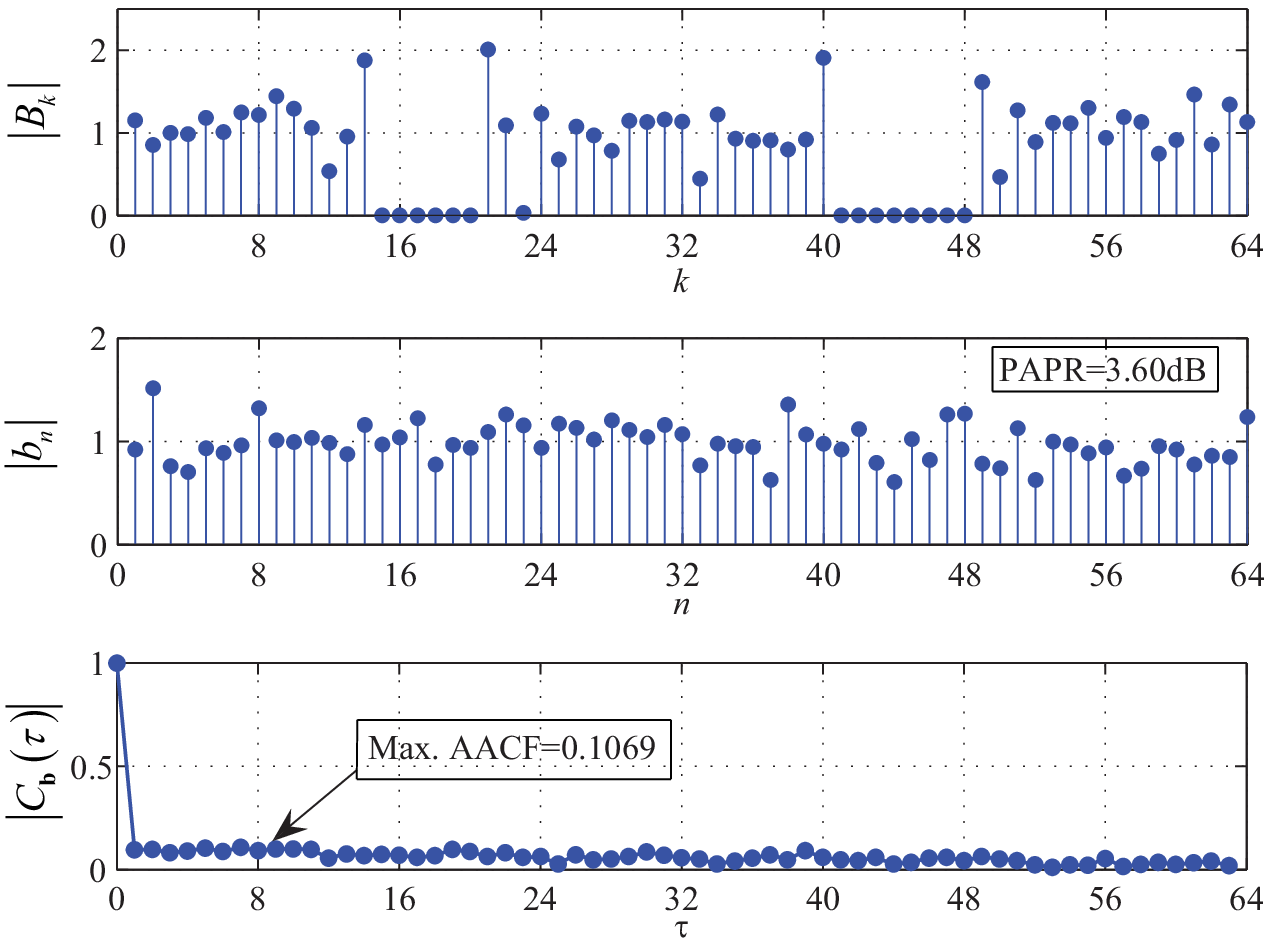}
\label{fig_examp3_b}}
}
\caption{A comparison of two CR sequences (with different $\lambda$) obtained by the proposed algorithm.}
\label{fig_TwoCR}
\end{figure*}

\begin{example}
Applying our proposed algorithm and using the same spectrum hole constraint in \textit{Example 2}, two CR sequences (which have penalty factors of $\lambda=0.15$ and $\lambda=0.95$, respectively) are shown in Appendix A. Similar to Fig. 3-a, the time- and frequency- domain magnitudes, and their AACF of these two CR sequences are shown in Fig. 5. It is seen that the CR sequence with
$\lambda=0.15$ displays almost flat time-domain magnitudes, giving rise to a PAPR of 1.10dB and maximum out-of-phase AACF magnitude (normalized) of 0.1377.
In contrast, when $\lambda$ is increased to 0.95, the maximum out-of-phase AACF magnitude (normalized) is reduced to 0.1069 accordingly, whereas its PAPR is increased to 3.60dB. This shows the tradeoffs between the PAPR and the AACF. Also, these two CR sequences are superior if compared with the CR sequence $\mathbf{b}_4$ in \textit{Example 2}, which has PAPR of 4.10dB and maximum out-of-phase AACF magnitude of 0.2131.
\end{example}

\vspace{0.1in}
\begin{example}
Consider the spectrum hole constraint below which is used in IEEE 802.11a to specify the nulled DC-subcarrier and guard bands.
\begin{equation}\label{spe_const_18}
\Omega = \{0\}\cup \{27,28,\cdots,37\}.
\end{equation}
Applying $\Omega$ in (\ref{spe_const_18}), we show in Fig. 6 a tradeoff comparison of our proposed algorithm and that in [\ref{Tsai11}, Section IV] for their capabilities in suppressing the PAPR and the AACF. To do this, we simulate the cases for $\lambda=i/1000$, where $i$ ranges over $[0,1,\cdots,1000]$. One can see that our proposed algorithm gives rise to a higher probability of lower PAPR and lower maximum out-of-phase AACF.
\end{example}

\vspace{0.1in}

\begin{remark}
Compared with the algorithms in \cite{He10}, an advantage of our proposed algorithm is that it leads to CR sequences with zero spectral leakage over the spectrum holes, as seen from the top two subplots for $|B_k|$ in Fig. 5. Also, compared with the algorithm in [\ref{Tsai11}, Section IV], our proposed algorithm is more effective in suppressing the PAPR and the AACF, as shown in \textit{Example 4}.
\end{remark}

\begin{figure}
\centering
  \includegraphics[width=3in]{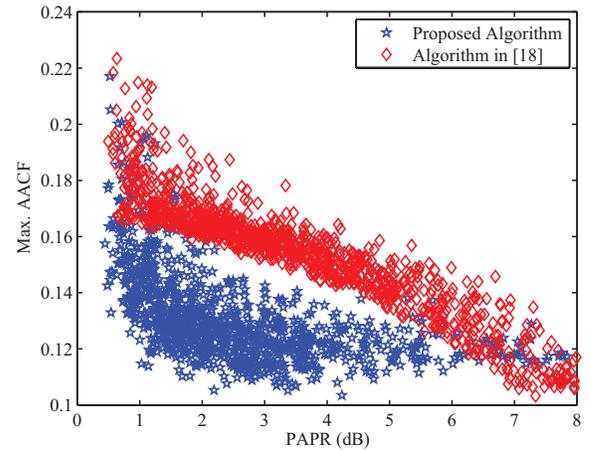}
  \caption{A comparison of the PAPR-AACF suppression capabilities between the proposed algorithm and the algorithm in [\ref{Tsai11}, Section IV].}
  \label{fig_PAPRvsAACF}
\end{figure}

\vspace{0.1in}

{Before we close this section, we present below some comments on the complexity of our proposed quasi-ZCZ CR sequence construction including the optimization algorithm.

\begin{remark}
Our proposed quasi-ZCZ CR sequence construction consists of two essential steps: seed sequence selection followed by numerical optimization. In the seed sequence selection step, a $(K,L,Z)$ ``seed" ZCZ sequence set is needed, but this is easy because there are abundant choices available in the literature for ZCZ sequences \cite{Tang01}\cite{Torii04}. Therefore, to construct our proposed quasi-ZCZ CR sequences of length $LN$, we have a $L$-times optimization complexity reduction over \cite{He10} and [\ref{Tsai11}, Section IV] because our overall complexity is mainly determined by the second numerical optimization step to be performed on a CR waveform of length $N$, while \cite{He10} and [\ref{Tsai11}, Section IV] need to perform numerical optimization over the entire length $LN$. In addition, compared to the algorithm in [\ref{Tsai11}, Section IV] which has to solve a relaxed non-convex optimization problem, our proposed algorithm is convex in nature and can be easily solved by some optimization tools, such as CVX, Matlab, etc. In short, the low complexity feature of our proposed quasi-ZCZ CR sequence construction will be useful in a CR channel with dynamic spectrum hole constraint.
\end{remark}}

\section{A CR-CDMA Scheme based on the Proposed Quasi-ZCZ Sequences}
In this section, we present a CR-CDMA scheme which uses the proposed quasi-ZCZ CR sequence set. 
Multipath fading channels are considered, each modeling as a discrete-time finite impulse response filter, i.e., $h[n]=\sum_{\tau=0}^{T_{\max}}h_\tau\delta[n-\tau]$, where $\delta[n]$ and $T_{\max}$ denote the discrete Dirac pulse and the maximum channel delay, respectively. For every transmitter in a CR-CDMA, a cyclic prefix (CP), which is a repetition of the last $T_g \geq T_{\max}$ samples in a data block\footnote{In this paper, a data block refers to a data symbol spread by a signature sequence in a CR-CDMA system.}, is inserted at the beginning of each data block, as shown in Fig.~\ref{FIG:BLKTX}. This is to maintain the periodic correlations of the proposed quasi-ZCZ CDMA sequences in quasi-synchronous multipath fading channels.

Let us consider a quasi-synchronous $K$-user CR-CDMA system where every user is assigned a specific spreading sequence constructed in (\ref{equ_propose_cons}). After passing through an asynchronous multipath fading channel, the received signal vector of the $i$th user after removing the CP is given by
\begin{equation}
\label{RXsignal}
{\bf{r}}_i = \sum \limits_{j=1}^{K} \sum \limits_{p=0}^{T_{max}} h_{p}^{i, j}  d_j  T^p (\mathbf{c}_j) +\mathbf{n}_i,
\end{equation}
where $h_{p}^{i, j}$ is the $p$th channel coefficient between the $i$th and $j$th users, and $d_j$ is the data symbol of the $j$th user. Suppose that the $i$th user is the user of interest in the receiver. Carrying out the despreading by the local reference sequence ${\bf{c}}_i$, we have
\begin{equation}\label{RXPACF2}
\begin{split}
<\mathbf{r}_i,\mathbf{c}_i> & =  \sum \limits_{j=1}^{K}  \sum \limits_{p=0}^{T_{\max}} h_{p}^{i,j} d_j \left <T^p(\mathbf{c}_j),\mathbf{c}_i \right > + \left <\mathbf{n}_i,\mathbf{c}_i \right > \\
                           & =\underbrace{\sum \limits_{p=0}^{T_{max}} h_{p}^{i, i} d_i R_{\mathbf{c}_i} \left( p \right)}_{\text{MPI}}\\
     &+  \underbrace{\sum \limits_{j=1, i \neq j}^{K}  \sum \limits_{p=0}^{T_{\max}} h_{p}^{i, j} d_j R_{\mathbf{c}_j,{\bf{c}}_i} \left( p \right)}_{\text{MUI}}+  \underbrace{ \left <\mathbf{n}_i,\mathbf{c}_i \right >}_{\text{Noise}},
\end{split}
\end{equation}
where the first term and the second term in the second step of (\ref{RXPACF2}) represent multipath interference (MPI) and MUI, respectively.

\vspace{0.1in}
\begin{remark}
Note that the MUI from the $j$th user can be completely eliminated if
\begin{equation}
R_{\mathbf{c}_j,{\bf{c}}_i} \left( p \right)=0,~~~ p \in [0, T_{\max}].
\end{equation}
Clearly, this can be ensured if we have
\begin{equation}\label{zeroMUI_cond}
T_{\max} \leq NZ-N,
\end{equation}
by recalling the cross-correlation property of the proposed quasi-ZCZ CR sequences shown in \eqref{NewCC}. In short, the proposed scheme provides MUI-free performance provided that the maximum channel delay is not greater than the ZCCZ width.
\end{remark}
\vspace{0.1in}

\begin{figure}
\centering
  \includegraphics[width=3.2in]{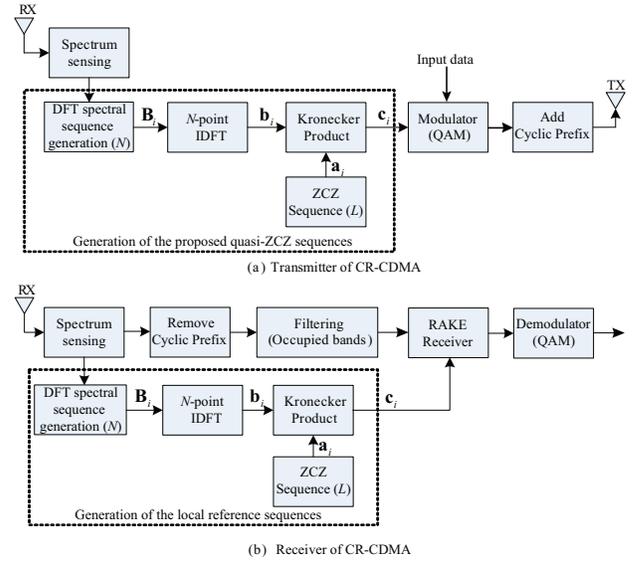}
  \caption{Block diagram of the proposed CR-CDMA system.}
  \label{FIG:BLKTX}
\end{figure}

Furthermore, when (\ref{zeroMUI_cond}) is satisfied, \eqref{RXPACF2} can be simplified to
\begin{equation}
\label{RXPACF3}
<\mathbf{r}_i,\mathbf{c}_i> =\underbrace{d_i \cdot \sum \limits_{p=0}^{T_{\max}} h_{p}^{i, i} R_{\mathbf{c}_i} \left( p \right)}_{\text{MPI}}+  \underbrace{ \left <\mathbf{n}_i,\mathbf{c}_i \right >}_{\text{Noise}}.
\end{equation}
As a result, the proposed scheme transforms the multiple access problem in CR-CDMA to a traditional single-user spread-spectrum problem over multipath fading channel. To exploit the multipath diversity, a maximal-ratio combining (MRC) based RAKE receiver is employed, as shown in Fig.~\ref{FIG:BLKTX}.

\section{Simulation Results}
In this section, we evaluate the proposed CR-CDMA system performance. The main objective is to show the interference-free achievability of the proposed CR-CDMA in quasi-synchronous cognitive radio network (CRN) with multipath fading channels. To this end, two other spreading-based CRNs are considered for comparison purpose, i.e., non-contiguous multicarrier-CDMA (MC-CDMA) \cite{Hara}, \cite{Boroujeny} and transform domain communication system (TDCS) \cite{Swackhammer}-\cite{QCCSK}. In particular, by adopting the baseband symbol modulation scheme called cyclic code shift keying (CCSK) \cite{Dillard}, a TDCS is capable of providing reliable communications with low spectral density using spectral nulling and frequency domain spreading. Simulation settings for the three CRNs (each consisting of four or more CR users) are shown in Table I.

\begin{table*}
\newcommand{\tabincell}[2]{\begin{tabular}{@{}#1@{}}#2\end{tabular}}
\label{table1}
\centering
\caption{Simulation Settings for Three Spreading-Based CRNs}
\begin{tabular}{|c||c|c|c|c|}
  \hline
   Parameters & MC-CDMA &  MC-CDMA & TDCS \cite{Chakravarthy-TDCS} & CR-CDMA (Proposed) \\
   \hline
   \hline
  Spreading codes & Zadoff-Chu & Polyphase Random & Polyphase Pseudo-Random & Quasi-ZCZ (Proposed) \\
  \hline
   Codes of length & 1024 &  1024 & 1024 &  1024 ($L=16, N=64$) \\
  \hline
  Orthogonality & No  &No & No & Yes \\
   \hline
  Modulation & QPSK & QPSK & CCSK & QPSK\\
  \hline
  Receiver & MMSE-FDE & MMSE-FDE & MMSE-FDE & MRC-RAKE \\
  \hline
  Channels & \multicolumn{4}{|c|}{COST207RAx6 \cite{Cost207}} \\
  \hline
  Entire bandwidth & \multicolumn{4}{|c|}{10MHz}\\
  \hline
  Unavailable bands & \multicolumn{4}{|c|}{2.5$\sim$3.75 MHz and 6.25$\sim$7.5 MHz}\\
  \hline
  \end{tabular}
\end{table*}

We also examine the near-far resilience of the three CRNs. In CDMA communications, the ``near-far effect" refers to a phenomenon that the desired user signal is overwhelmed by the interfering user signals from within a region which is closer to the receiver. We will show that the proposed quasi-ZCZ CR sequences give rise to ``near-far immune" CR-CDMA systems, in contrast to ``near-far sensitive" MC-CDMA and TDCS.

Let $E_k$ ($1\leq k\leq K$) be the received signal energy of the $k$th user. Denote by ${NF}_{i,j}=E_j/E_i$ the "near-far factor" between the $j$th user (an interfering user) and the $i$th user (the desired user), where $i\neq j$. In addition, whenever the subscript is omitted, unless other specified, a near-far factor of $NF$ means that all interfering users have identical signal energy of $E_iNF$.

\subsection{BER Performance in Quasi-Synchronous Multipath Fading Channels}
To evaluate the BER performance in quasi-synchronous multipath fading channels, a CP length of $1/4$ (relative to the total data block duration) is adopted for the above three CRNs. For a MC-CDMA system, the one-tap equalizer, i.e., minimal mean square error frequency-domain equalization (MMSE-FDE), is used in the receiver. Zadoff-Chu sequences and polyphase random sequences are considered. For a CR-CDMA system, MRC-RAKE receiver is used to exploit the multipath diversity. It is seen from Fig.~\ref{fig_BER_sim} that the single-user BER performance is achieved by the proposed CR-CDMA for $NF=10$dB, indicating that the proposed CR-CDMA is near-far immune. In contrast, the BER gap between MC-CDMA (or TDCS) and the single-user case grows larger for increasing SNR (denoted by $E_b/N_0$) because in this case the MUI becomes dominant. This result shows that both MC-CDMA and TDCS are near-far sensitive owing to the non-zero MUI. Note that for MC-CDMA, the BER curves using Zadoff-Chu sequences and polyphase random sequences are very close. This is because the spectral nulling destroys the orthogonality of Zadff-Chu sequences, resulting in a correlation property like that of a random sequence. {It is noted that our proposed CR-CDMA system features an MRC-RAKE receiver structure and thus it is capable of achieving multipath diversity which cannot be achieved by an (uncoded) OFDM system. Specifically, the average BER is determined by $n$-order diversity combining (where $n$ is the number of uncorrelated channel paths) \cite{CDMA-handBook} and is inversely proportional to the MPI determined by the sum of multipath tap gain weighted by the PACF sidelobe [as shown in (\ref{RXPACF3})].}

\begin{figure}
\centering
  \includegraphics[width=3.2in]{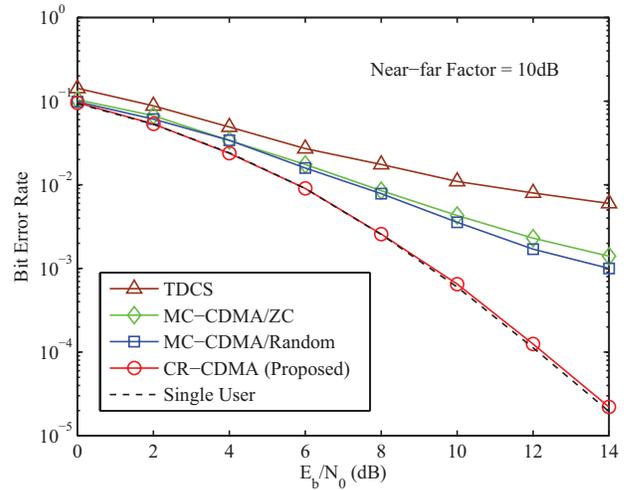}
  \caption{BER performance of CR-CDMA (proposed), MC-CDMA and TDCS in quasi-synchronous multipath fading channels}
  \label{fig_BER_sim}
\end{figure}

\subsection{BER Performance versus MUI}
In this part, we evaluate the BER performance versus MUI. Fig.~\ref{fig_BERMUI}-a shows the effect of different near-far factors on the BER performances of the three spreading-based CRNs with $K$=4 and the near-far factor $NF$ ranging from 0dB to 20dB. It is seen that the BER performance of MC-CDMA (or TDCS) degrades rapidly as $NF$ increases. For example, the BER performance of MC-CDMA degrades from $10^{-3}$ to $10^{-1}$ when $NF$ changes from 0dB to 20dB. In contrast, the BER curve of the proposed CR-CDMA agrees very well with the single-user one even at a large $NF$ region, further verifying the near-far immunity of CR-CDMA.

{Fig.~\ref{fig_BERMUI}-b demonstrates the effect of number of users to the BER performance, where $NF$=10dB. {We employ a bound-achieving ``seed" ZCZ sequence set \cite{Popovic10} with $K=16$} so as to support up to 16 users. In addition, a practical CR channel scenario with four spectrum holes and $50\%$ available subcarriers is considered. It is clear that the BER performance of the proposed CR-CDMA is insensitive to the number of users due to zero cross-correlation in the zone, whereas both the MC-CDMA and TDCS systems suffer from increased MUI (leading to severe BER degradation) for increasing number of users}. For practical decentralized CRNs, such as CR Ad-hoc wireless sensing networks where near-far problem commonly exists, our proposed CR-CDMA is superior as it is near-far immune and does not require a costly power control loop.

\begin{figure*}
\centerline{
\subfloat[Near-Far Effect]{\includegraphics[width=3.2in]{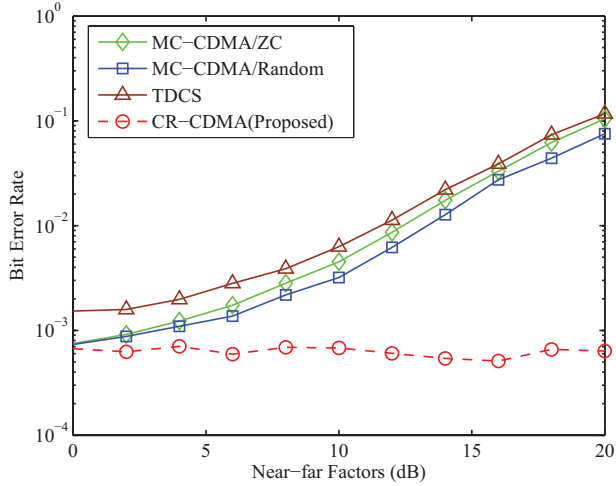}
}
\hfil
\subfloat[Number of Users]{\includegraphics[width=3.2in]{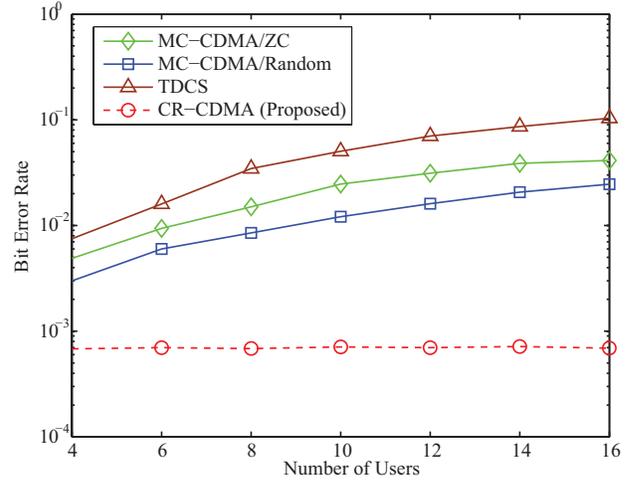}
}
}
\caption{BER performance versus multiuser interference for CR-CDMA (proposed), MC-CDMA and TDCS, where $E_b/N_0 = 10$dB.}
\label{fig_BERMUI}
\end{figure*}

\subsection{Spectrum Sensing Mismatch}
Denote by
\begin{displaymath}
\begin{split}
{\mathbf{S}}_{T} & =\left[ S_0^{T}, S_1^{T}, \cdots, S_{N-1}^{T} \right]^{\text{T}},\\
{\mathbf{S}}_{R} & =\left[ S_0^{R}, S_1^{R}, \cdots, S_{N-1}^{R} \right]^{\text{T}},
\end{split}
\end{displaymath}
 the subcarrier marking
vectors detected by the transmitter and the receiver, respectively. In practice, channel uncertainties may lead to spectrum sensing mismatch between transmitter and receiver \cite{Cost207}, i.e., ${\mathbf{S}}_{T}\neq {\mathbf{S}}_{R}$. To quantify the spectrum sensing mismatch, we define the correlation coefficient $\eta$ \cite{Martin} as follows,
\begin{equation}
\eta  = \frac{{{\mathbf{S}}^{\text{T}}_{T}{\mathbf{S}}_{R} }} {{\left\|
{{\mathbf{S}}_{T} } \right\|\left\| {{\mathbf{S}}_{R} } \right\|}}.
\end{equation}
In particular, the perfect spectrum sensing is referred to if $\eta=1$, i.e., ${\mathbf{S}}_{T} = {\mathbf{S}}_{R}$. Otherwise, the spectrum sensing mismatch occurs.

{Fig.~\ref{FIG:BERMISMATCH} shows the BER performances of CR-CDMA for $\eta\in \left\{87\%,89\%,92\%,96\%\right\}$.
Specifically, for $\eta\in \left\{87\%,89\%\right\}$, we used a CR channel setting with four spectrum holes and $50\%$ available subcarriers, whereas for $\eta\in \left\{92\%,96\%\right\}$, we used two spectrum holes and $75\%$ available subcarriers. For analysis, we study the $E_b/N_0$ loss (compared to that of the perfect spectrum sensing case with $\eta=100\%$) at BER = $10^{-3}$. One can see that the BER curves for $\eta=96\%$ and $92\%$ are very close to that of the perfect spectrum sensing, with only $0.15$dB and $0.4$dB $E_b/N_0$ loss, respectively. As expected, the BER performance is degraded gradually as $\eta$ decreases. This is because smaller $\eta$ leads to less effective signal power captured by the receiver. For instance, there is $2$dB $E_b/N_0$ loss in the case of $\eta=87\%$. In contrast, when non-continuous OFDM (NC-OFDM) \cite{Budiarjo} is used, the BER curves under the same $\eta$ settings show very high error floors at all $E_b/N_0$ levels, indicating that it is very sensitive to spectrum sensing mismatch. Therefore, our proposed CR-CDMA system is more robust against modest spectrum sensing mismatch.}

\begin{figure}
\centering
  \includegraphics[width=3.2in]{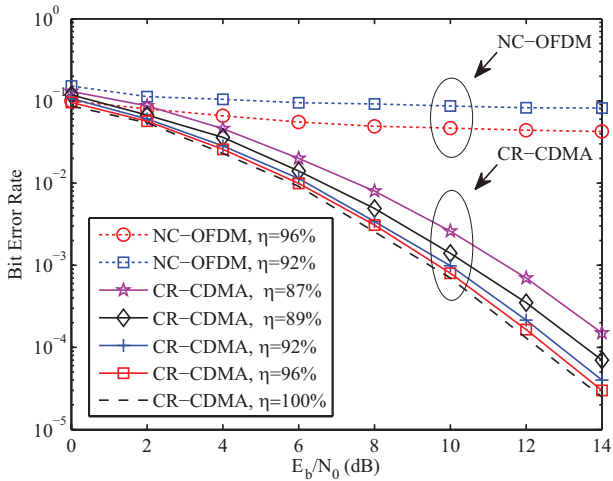}
  \caption{BER comparison of CR-CDMA and NC-OFDM in the case of spectrum sensing mismatch}
  \label{FIG:BERMISMATCH}
\end{figure}

\section{Conclusions}
To achieve interference-free CR-CDMA communications, a design barrier is the ¡°spectrum hole constraint¡± in the CR channel, which renders traditional ZCZ sequences unusable because the ZCZ correlation property of the ZCZ sequences will be destroyed by spectral nulling. Motivated by this, we have presented in Section III.A the systematic construction of a novel $(K,NL,NZ-N)$ quasi-ZCZ CR sequence set with set size $K$, length $NL$, and whose every pair of distinct sequences has zero cross-correlation zone (ZCCZ) width of $NZ-N$. The key idea in the proposed construction is to perform the Kronecker product of a $(K,L,Z)$ ZCZ sequence set\footnote{where $K$ denotes the set size, $L$ the sequence length, and $Z$ the ZCZ width.} and a waveform set (each of length $N$) which satisfies spectrum hole constraint. Due to the property of Kronecker product, the ZCCZ property of our proposed CR sequences is guaranteed to hold for any spectrum hole constraint (which may vary with time and location).

We have also presented a numerical optimization algorithm to further optimize the proposed CR sequences to have low PAPR and low AACF simultaneously. Our optimization algorithm takes advantage of the property that a sequence with low AACF also possesses low PACF. This leads to simplified convex optimization, instead of the non-convex optimization which is necessary in [\ref{Tsai11}, Section IV].

The proposed CR sequences possess zero cross-correlation zone and near-zero auto-correlation zone property, hence they are called ``quasi-ZCZ".  Computer simulations of these quasi-ZCZ sequences show that they are effective in achieving MUI-free and near-far resistant CR-CDMA system in quasi-synchronous multipath fading channels. Compared to non-contiguous OFDM which does not perform spreading, the proposed CR-CDMA is also more robust to spectrum sensing mismatch. {Possible future works of this research include:
\begin{enumerate}
\item No channel coding is considered in this paper. Channel coding is expected to reduce the spreading factor (sequence length) of CR-CDMA system but provide error correction. It will be interesting to investigate the optimum [code rate, spreading factor] pair of the proposed CR-CDMA system under different channel conditions.
\item A recent hot research topic in sequence design is Golay complementary pairs (GCPs) \cite{GOLAY61},\cite{Liu-TIT2013} and complementary sequences \cite{TSENG72},\cite{Liu-TCOM2014}. To design complementary sequence sets with large set size, ``quasi-complementary sequence sets (QCSS)" with low correlations have been proposed \cite{Liu-TCOM2011}-\cite{Liu-IT2013}. It will be interesting if ``complementary CR sequences" can be designed. Here, ``complementary CR sequences" refer to a set of two-dimensional matrices with zero/low correlation sums when subjected to identical row nulling (corresponding to spectral nulling).  Compared to complex-valued CR sequences, we expect the complementary CR sequences to have a smaller alphabet size.
\end{enumerate}
}
\vspace{0.2in}

\appendices

\begin{figure*}
\begin{enumerate}
\item $\lambda=0.15$
\begin{equation}
 ~~~~|\mathbf{b}| =\left [ \begin{array}{l}
    1.0844,    0.9689,    0.9512,    0.8655,    1.0119,    1.0310,    0.9329,    0.8784,    0.9692,1.0832\\
    0.9306,    1.1345,    1.0252,    0.9561,    0.9639,    0.9628,    0.8833,    1.0476,    1.1166,0.9487\\
    1.0493,    1.0227,    0.9900,    0.9454,    1.0702,    0.9735,    1.1100,    1.0073,    1.0217,1.0184\\
    1.0264,    1.0203,    1.0336,    1.0249,    0.8444,    1.0486,    1.0715,    0.9755,    0.9327,1.0599\\
    0.9405,    1.0868,    1.0331,    0.9506,    1.0494,    1.0802,    0.8774,    0.9737,    0.9641,1.1247\\
    1.0245,    1.1301,    0.9918,    0.8887,    1.0272,    0.9673,    0.9471,    1.0037,    1.0594,1.0301\\
    0.8108,    0.9433,    0.9657,    0.9772
\end{array} \right ]
\end{equation}

\begin{equation}
  \arg\{\mathbf{b}\}  = \left [ \begin{array}{l}
    3.5907,    5.0480,    4.9631,    1.0373,    1.2193,    0.4437,    0.1947,    1.6065,    0.4091,1.9281\\
    3.2030,    4.0366,    0.2956,    5.0776,    5.6550,    2.6589,    1.1841,    1.6763,    0.0331,3.8943\\
    4.2418,    4.4046,    2.0934,    1.1538,    2.2643,    4.2935,    3.7441,    4.2361,    6.2214,5.4689\\
    3.5310,    4.9327,    1.8726,    4.9712,    5.6652,    2.6973,    0.5947,    1.6379,    0.7703,2.5467\\
    2.1298,    1.0033,    1.5913,    1.1897,    0.0287,    0.2438,    5.4877,    5.1932,    1.7509,4.0920\\
    5.5475,    0.9042,    5.0590,    1.1550,    4.6100,    2.7825,    5.7043,    0.8719,    1.5994,6.0930\\
    0.4790,    6.2167,    3.6842,    1.5993
\end{array} \right ]\\
\end{equation}

\item $\lambda=0.95$
\begin{equation}
~~~~|\mathbf{b}| = \left [ \begin{array}{l}
    0.9202,    1.5131,    0.7567,    0.7006,    0.9339,    0.8890,    0.9586,    1.3201,    1.0091,    0.9961\\
    1.0343,    0.9875,    0.8789,    1.1583,    0.9699,    1.0377,    1.2240,    0.7759,    0.9650,    0.9370\\
    1.0886,    1.2633,    1.1559,    0.9409,    1.1706,    1.1339,    1.0144,    1.2027,    1.1104,    1.0400\\
    1.1571,    1.0714,    0.7685,    0.9796,    0.9505,    0.9447,    0.6260,    1.3551,    1.0689,    0.9795\\
    0.9210,    1.1189,    0.7910,    0.6055,    1.0201,    0.8188,    1.2632,    1.2663,    0.7815,    0.7412\\
    1.1276,    0.6259,    0.9985,    0.9709,    0.8844,    0.9435,    0.6671,    0.7342,    0.9525,    0.9222\\
    0.7751,    0.8616,    0.8485,    1.2352
\end{array} \right ]\\
\end{equation}

\begin{equation}
  \arg\{\mathbf{b}\} = \left [ \begin{array}{l}
    4.8171,    3.2030,    0.4232,    1.1852,    4.4278,    2.9172,    5.0056,    0.9800,    4.2892,    3.8035\\
    3.2985,    4.0981,    4.7946,    1.9267,    3.5038,    1.4390,    4.9246,    2.8596,    3.6866,    3.6402\\
    3.2894,    4.8848,    1.9006,    1.3994,    1.9457,    2.4315,    5.2810,    3.5875,    4.7004,    3.9169\\
    0.7822,    5.1472,    0.3984,    2.7239,    3.2930,    5.7733,    5.1629,    5.5392,    1.2946,    0.2351\\
    1.4322,    6.2033,    0.1863,    0.3297,    4.0488,    2.9717,    2.5382,    3.4520,    4.7880,    3.5623\\
    4.3384,    0.5327,    2.2144,    2.6802,    0.7534,    0.4189,    0.7884,    5.6823,    2.3551,    3.9450\\
    4.7848,    0.5745,    5.8238,    4.5391
\end{array} \right ]\\
\end{equation}

\end{enumerate}
\end{figure*}


\section{Two CR sequences obtained from our proposed algorithm}
For ease of presentation, denote by $|\mathbf{b}|$ and $\arg\{\mathbf{b}\}$ the magnitude vector and the phase vector (in radian) of $\mathbf{b}$, respectively.
Also, each vector is arranged in a matrix form. Therefore, to get $|\mathbf{b}|$ (which is a column vector), for instance, just read out its corresponding matrix row by row, followed by the transpose operation.

\section*{Acknowledgment}
The authors would like to thank Shu FANG, Yue XIAO and Gang WU for many useful comments on cognitive radio code division multiple access.


\end{document}